\documentclass[11pt,fleqn]{article}

\usepackage{booktabs} %
\usepackage[ruled]{algorithm2e} %

\usepackage{amsthm}

\usepackage{amssymb}
\usepackage{graphicx}
\usepackage{aligned-overset}
\usepackage{dsfont}

\usepackage{xspace}
\usepackage{enumitem}
\usepackage{tabularx}
\usepackage{xcolor}
\usepackage{graphicx} 
\usepackage{tikz}
\usepackage{wrapfig}

\usepackage{thmtools} %

\usepackage{hyperref}
\hypersetup{colorlinks=true,citecolor=red}
\usepackage{cleveref}

\newcommand{\etal}{et al.\ }

\newcommand{\R}{\mathbb{R}}

\newtheorem{theorem}{Theorem}[section]
\newtheorem{lemma}[theorem]{Lemma}

\newtheorem{observation}[theorem]{Observation}

\newtheorem{corollary}[theorem]{Corollary}

\newtheorem{claim}[theorem]{Claim}

\newtheorem{remark}[theorem]{Remark}
\newtheorem{example}{Example}

\def\+#1{\mathcal{#1}}

\newcommand{\opt}{\mathsf{opt}}

\newcommand{\norm}[1]{\lVert#1\rVert}
\newcommand{\abs}[1]{\AutoAdjust{\lvert}{#1}{\rvert}}

\newcommand{\be}{\begin{equation}}
\newcommand{\ee}{\end{equation}}
\newcommand{\beq}{\begin{equation*}}
\newcommand{\eeq}{\end{equation*}}

\newcommand{\eps}{\varepsilon}

\newcommand{\AutoAdjust}[3]{\mathchoice{ \left #1 #2  \right #3}{#1 #2 #3}{#1 #2 #3}{#1 #2 #3} }
\newcommand{\Xcomment}[1]{{}}

\newcommand{\InBrackets}[1]{\AutoAdjust{[}{#1}{]}}%
\newcommand{\Ex}[2][]{\operatorname{\mathbb E}_{#1}\InBrackets{#2}}
\newcommand{\Exlong}[2][]{\operatornamewithlimits{\mathbb E}\limits_{#1}\InBrackets{#2}}
\newcommand{\Prx}[2][]{\operatorname{\mathbf{Pr}}_{#1}\InBrackets{#2}}

\newcommand{\eqdef}{\overset{\mathrm{def}}{=\mathrel{\mkern-3mu}=}}
\newcommand{\vect}[1]{\ensuremath{\mathbf{#1}}}

\newcommand{\metric}{\mathcal{M}}
\newcommand{\elli}[1][i]{\ell_{#1}}
\newcommand{\ells}{\ensuremath{\boldsymbol\ell}}
\newcommand{\ellsmi}[1][i]{\ells_{\text{-}#1}}
\DeclareMathOperator{\SC}{\textsf{SC}}
\DeclareMathOperator{\cost}{\textsf{cost}}
\newcommand{\costi}[1][i]{\cost_{#1}}

\newcommand{\shortpath}{\textsf{S-Path}}

\newcommand{\locs}{\ensuremath{\boldsymbol\ell}}
\newcommand{\locsmi}[1][i]{\ensuremath{\boldsymbol\ell}_{\text{-}#1}}
\newcommand{\rli}[1][i]{\widetilde{\ell}_{#1}}  
\newcommand{\rlj}[1][j]{\widetilde{\ell}_{#1}}  %
\newcommand{\rls}{\widetilde{\boldsymbol{\ell}}}
\newcommand{\rlsmi}[1][i]{\widetilde{\boldsymbol{\ell}}_{\text{-}#1}}
\newcommand{\po}{\hat o}  %
\newcommand{\pa}{\hat a}  %
\newcommand{\pac}{a}  %
\newcommand{\pam}{\Delta}  %
\newcommand{\harm}{\texttt{Harmonic}}

\newcommand{\dictator}{\texttt{RD}}

\newcommand{\tdi}[1][i]{t_{#1}}

\newcommand{\denom}{D}
\newcommand{\denomr}{\widetilde{D}}

\newcommand{\distsmi}[1][i]{\mathcal{F}_{\text{-}#1}}
\newcommand{\br}{\textsf{BR}}
\newcommand{\bri}[1][i]{\br_{#1}}

\definecolor{WildStrawberry}{RGB}{255,67,164}

\def\DEBUG{true}
\ifdefined\DEBUG
	
	\def\rem#1{{\marginpar{\raggedright\scriptsize #1}}}
	\newcommand{\sjr}[1]{\rem{\small\textcolor{red}{$\bullet${\tiny #1}}}}
	
	\newcommand{\remove}[1]{{\color{lightgray} #1}}
\else
	
	\newcommand{\sjr}[1]{}
	\new
	command{\remove}[1]{}
\fi
\def\DEBUG{true}
\ifdefined\DEBUG
	
	\def\rem#1{{\marginpar{\raggedright\scriptsize #1}}}
	\newcommand{\msr}[1]{\rem{\small\textcolor{red}{$\bullet${\tiny #1}}}}

\else
	
	\newcommand{\msr}[1]{}
	\new
	command{\remove}[1]{}
\fi

\begin{document}

\sloppy

\title{Strategic Facility Location via Predictions}

\author{
Qingyun Chen \thanks{ Electrical Engineering and Computer Science, University of California, 5200 N. Lake Road, Merced CA 95344. \texttt{qchen41@ucmerced.edu}.}  
\and 
Nick Gravin\thanks{ Key Laboratory of Interdisciplinary Research of Computation and Economics, Shanghai University of Finance and Economics, Shanghai, China. \texttt{anikolai@mail.shufe.edu.cn}.
} 
\and
Sungjin Im\thanks{ Electrical Engineering and Computer Science, University of California, 5200 N. Lake Road, Merced CA 95344. \texttt{sim3@ucmerced.edu}.}  
}

\date{}
\maketitle

\begin{abstract}

The facility location with strategic agents is a canonical problem in the literature on mechanism design without money. Recently, 
Agrawal et. al.~\cite{agrawal2022learning} considered this problem
in the context of machine learning augmented algorithms, where the mechanism designer is also given a prediction of the optimal facility location. An ideal mechanism in this framework produces an outcome that is close to the social optimum when the prediction is accurate (consistency) and gracefully degrades as the prediction deviates from the truth, while retaining some of the worst-case approximation guarantees (robustness). The previous work only addressed this problem in the two-dimensional Euclidean space providing optimal trade-offs between robustness and consistency guarantees for deterministic mechanisms. 

We consider the problem for \emph{general} metric spaces. Our only assumption is that the metric is continuous, meaning that any pair of points must be connected by a continuous shortest path. We introduce a novel mechanism that in addition to agents' reported locations takes a predicted optimal facility location $\hat{o}$. We call this mechanism $\texttt{Harmonic}$, as it selects one of the reported locations $\tilde{\ell}_i$ with probability inversely proportional to 
$d(\hat{o},\tilde{\ell}_i)+ \Delta$ for a constant parameter $\Delta$. While \harm \ mechanism is not truthful, we can \emph{characterize the set of undominated strategies} for each agent $i$ as solely consisting of the points on a shortest path from their true location $\ell_i$ to the predicted location $\hat{o}$. We further derive \emph{consistency and robustness guarantees on the Price of Anarchy (PoA)} for the game induced by the mechanism. Specifically, assuming that $\Delta=\frac{\varepsilon}{2n}\cdot\sum_{i\in [n]}d(\hat{o},\ell_i)$ is closely related to the average distance to the predicted location $\hat{o}$, our consistency guarantee is arbitrarily close to optimum (PoA is $1+\varepsilon$), while having a constant robustness guarantee (PoA is $O(1+1/\varepsilon^3)$ in general and is $O(1+1/\varepsilon^2)$ for strictly convex metric spaces). We also show that for a constant number of agents $n=O(1)$, $\harm(\Delta)$ mechanism with $\Delta=0$ attains $1$-consistency and $O(1)$-robustness.
\end{abstract}

\section{Introduction}
\label{sec:intro}

The facility location is a canonical problem attracting a lot of interest
in many different fields such as operation research, artificial intelligence, social choice, and economics~\cite{chan2021mechanism}. Facility location with strategic customers serves as an exemplary setting of the approximate mechanism design without payments that has been extensively studied in the past fifteen years~\cite{procaccia2013approximate,AlonFPT10,LuSWZ10,EscoffierGTPS11,FeldmanW13,FotakisT10,FotakisT14,FotakisT16,SerafinoV16,procaccia2018approximation,walsh2020strategy,AzizLSW22}.

In the case of a single facility, the problem involves $n$ agents 
residing in a metric space $\metric$, each agent $i$ with their own most preferred location $\elli\in\metric$ of the single facility $f\in\metric$. The agents are strategic and may misreport their location hoping to influence the final choice of the facility $f\in\metric$ and minimize their 
personal cost given by the distance $d(\elli,f)$. The central planner collects reported locations from the agents and then decides where to open the facility. For the most common objective, a.k.a., social cost minimization, the goal is two-fold: first, the outcome should be a good approximation for minimizing the average distance from the facility to all agents if not the optimum; and second, no agents should have an incentive to misreport their true location.

The traditional approach to this problem aims to find mechanisms with worst-case performance guarantees on all possible inputs. This approach has been criticized in a broader context of algorithm design for its often pessimistic bounds and its focus on adversarial instances, which are not frequently encountered in practice~\cite{roughgarden2021beyond}. 
To overcome the practical shortcomings of  traditional algorithms, Lykouris and Vassilvitskii~\cite{LykourisV21}
have recently introduced a powerful framework of
\emph{machine learning augmented algorithms}.  ML augmented algorithms strive to surpass the limitations of traditional algorithms by incorporating predictions as supplementary input. These algorithms are supposed to do very well when furnished with accurate predictions (consistency guarantee), while gracefully degrading in performance as prediction quality declines. Ideally, they also uphold robust guarantees for all predictions, ensuring reliability even in the face of imperfect forecasts (robustness guarantee). ML augmented algorithms have found applications in strengthening online and streaming algorithms with predictions about future unknown input, as well as in accelerating algorithms by utilizing knowledge gleaned from past inputs. See 
 ~\cite{MitzenmacherV22} for a survey  of ML augmented algorithms.

 Very recently, there has been a surge of interest in applying ML augmented framework to mechanism design~\cite{XuL22,GkatzelisKST22,BalkanskiGT23,mechanismdesignaugmentedoutput,rand-facility-game} pioneered by Agrawal et al.~\cite{agrawal2022learning}.
Their work
specifically addresses the strategic facility location problem in two-dimensional Euclidean space $\metric=\R^2$.  To illustrate the setting, consider the following examples:
\begin{example}
    \label{ex:1}
 A group of $n=100$ people (say alumni from the same school) want to decide on a  meeting place (e.g., for the school reunion). Before asking people about their preferred meeting location, the meeting planner may already have a good idea about how many people still live near the school and how many 
will be coming from a greater city area; thus the planner could predict which spot might work well for everyone (either somewhere close to the school, or a well accessible place in the city center). \vspace{-2mm}
\end{example}

\begin{example}
A city wants to open a new government office to serve the public. The city likely possesses extensive data on residents, thus can find a good candidate location. 
Alternatively, it might opt for a central location as a likely near-optimal solution, even without in-depth analysis.
\end{example}

These kinds of scenarios are neatly captured by the model of~\cite{agrawal2022learning}, which assumes access to 
a prediction of the optimal location obtained via machine-learning algorithm, or other means. This assumption about the output (and not the input)\footnote{It is common in the ML augmented algorithms literature to have predictions about algorithm's input rather than output.} is particularly well suited for the mechanism design scenario in which it is usually difficult to accurately predict the input from all $n$ agents.

Agrawal et al.~\cite{agrawal2022learning} demonstrate that it is possible to design a strategy-proof mechanism that is arbitrarily close to the optimal solution in the two-dimensional setting, provided access to a quality prediction of the facility location. This is a significant advancement, as previously only a $\sqrt 2$-approximate strategy-proof mechanism was known for this setting without predictions. Their mechanism builds upon a well known coordinate-wise median-point mechanism for low dimensional Euclidean spaces $\R^d$: the predicted location $\po\in\R^2$ gets a certain weight $c\in[0,1]$  and is added to the agents' reports, the output is given by the coordinate-wise median. The parameter $c$ is based on the expected quality of predictions: if the mechanism designer anticipates high-quality predictions, they can set the parameter $c$ close to 1, resulting in a nearly 1-consistent (nearly optimal) algorithm.
Specifically, the algorithm is $(1+r)$-
consistent and $g(r)$-robust\footnote{In general, an algorithm is $\alpha$-consistent and $\beta$-robust if it is an $\alpha$-approximation when the prediction is perfect and a $\beta$-approximation for arbitrary predictions.}, where $r(c) \in (0, \infty)$ is a control-parameter that governs a trade-off between consistency and robustness and $g(r)\approx\frac{1}{\sqrt{2r}}$ is a function tending to infinity as $r \rightarrow 0$. I.e., when prediction is extremely inaccurate, the algorithm maintains a constant approximation guarantee for any fixed value of $r$, ensuring a robust solution quality. These consistency and robustness guarantees are best possible among all anonymous deterministic mechanisms~\cite{agrawal2022learning}.

In this paper, we seek to investigate if the facility location problem continues to benefit from predictions even in the \emph{general} metric setting. We ask 

\begin{quote}
\begin{center}
    \emph{
    How to use ML augmented advice in general metric spaces $\metric$?}
\end{center}
\end{quote}

\subsection{Challenges}

We discuss two primary approaches considered in the literature and the difficulties in extending them to the general metric space with predictions. 

\smallskip
\noindent \textbf{Median-point like mechanisms.}
The median-point mechanism is perhaps the most common approach for the facility location game. It assumes that the locations are in the (low-dimensional) Euclidean space and computes the median of the reported locations along each coordinate. The algorithm is strategyproof, and it is optimum for the single dimensional Euclidean space~\cite{procaccia2013approximate} and $\sqrt 2$-approximate for the two dimensional Euclidean space \cite{goel2023optimality}. 
However, the only known approximation guarantee of the coordinate-wise median is $O(\sqrt{d})$ for the Euclidean norm in $\R^d$ spaces, and the mechanism is inapplicable in general metric spaces.
In fact, it is known that any deterministic mechanism is $\Omega(n)$ factor off the optimum even in a circle metric~\cite{schummer2002strategy}. Hence, the approach of Agrawal et al.~\cite{agrawal2022learning} is only limited to low dimensional Euclidean spaces and does not extend to general metrics.

\smallskip
\noindent 
\textbf{Random Dictatorship.}
Understanding mechanism design in general metric spaces presents significant challenges, as evidenced by the limited research conducted in this area (see a survey~\cite{chan2021mechanism}). Indeed, any mechanism with a sublinear in $n$ approximation guarantee must be randomized~\cite{schummer2002strategy}. The only known $O(1)$-approximation strategyproof mechanism that works for general $\metric$ is \emph{Random Dictatorship} ($\dictator$), which chooses one of the reported locations uniformly at random. $\dictator$ is a $2(1 - 1/n)$-approximation~\cite{alon2009strategyproof}.
There are a few characterization results of strategyproof mechanisms for special metrics $\metric$ (such as trees, or cycles)~\cite{DokowFMN12,MeirLLMK21}. Also there is a characterization of group-strategyproof mechanisms~\cite{TangYZ20} in strictly convex spaces, which 
can only give $\Omega(n)$ approximation to the optimal 
social cost. Unfortunately, these results do not offer any new mechanisms that can work in general metric spaces. 
On the positive side, the ML augmented framework offers a richer space of possible mechanisms by adding extra information about the predicted location to the input. On the negative side, there are still fundamental roadblocks to having good strategyproof ML augmented mechanisms, which are discussed below.

\smallskip
\noindent 
\textbf{Possible Mechanisms.}
First, it is reasonable to restrict attention to randomized mechanisms that select facility $f$ 
from a distribution over reported locations $\ells$ and the predicted location $\po$, as we may not even know which other points general $\metric$ may or may not have. This is still a rich family of mechanisms as each selection probability $\{g_i(\cdot)\}_{i\in[n+1]}$ may depend on $O(n^2)$ pairwise distances. Unfortunately, we do not know any strategyproof mechanisms of such form beyond $\dictator$ and for a good reason: each agent controls $n$ variables with many degrees of freedom in general $\metric$;  manipulations by a single agent may affect all $g_j$ for $j\in[n+1]$ in a rather complex way even for simple and well-behaved functions $\{g_i(\cdot)\}_{i\in[n+1]}$. One can still have a strategyproof mechanism by choosing $f$ from a fixed distribution over reported locations $\ells$ and predicted optimal location $\po$. However, if $\Prx{f\gets \po}$ is not a constant independent of $n$, then we do not get any interesting consistency guarantees. On the other hand, if $\po$ is very far from each reported location $\elli$, then we do not get any robustness guarantee.

A more reasonable attempt to balance consistency and robustness is to combine $\dictator$ mechanism with a selection of $\po$ by letting agents vote on two alternative outcomes: $\dictator(\ells)$ or $\po$. This approach, however, has its own problems that can be illustrated with our previous example of finding a meeting place for $n=100$ alumni (Example~\ref{ex:1}). Imagine that a single person $i^*\in [n]$ is on vacation at a remote tropical island in another country at the time of the meeting, while the remaining $99$ people are all in 
town. Although the $\dictator$ mechanism selects $i^*$ only with a 
small $1\%$ probability, it is still incredibly 
inconvenient for anyone else. We can think of $i^*$ as an adversarial agent, who does not care about the outcome, but still has a 
unilateral power to sabotage the outcome of $\dictator(\ells)$ for everyone else and make it arbitrarily worse than the predicted location 
$\po$. Thus any approach involving voting on two alternatives $\dictator(\ells)$ or $\po$ must give each agent a unilateral power to impose $\po$ outcome. This leads to  \emph{Consensus Mechanism}: the outcome is a predicted location $\po$, unless all agents agree on $\dictator(\ells)$ outcome. The consensus mechanism has perfect $1$-consistency, since at least one agent must prefer $\po$ over $\dictator(\ells)$ when $\po$ is the optimal location. It also has a constant robustness guarantee for a small number of agents $n=O(1)$ (e.g., $n=3$), since (i) in the case when at least one agent prefers $\po$ over $\dictator(\ells)$, the prediction is a constant approximation to the optimum, and otherwise (ii) $\dictator(\ells)$ is a $2(1-1/n)$-approximation to the optimum~\cite{alon2009strategyproof}.

\smallskip
\noindent 
\textbf{Untruthful Mechanisms.} It turns out that the consensus mechanism is not strategyproof. E.g., if one agent $i$ slightly prefers $\po$ over $\dictator(\ells)$ and everyone else strongly prefers $\dictator(\ells)$, then one of the agents $j\ne i$ may greatly benefit by reporting a closer location $\rlj$ to $\elli$ and thus convince $i$ to choose $\dictator(\ellsmi[j],\rlj)$ rather than $\po$. On the other hand, any selfish deviation $\rli\ne\elli$ by agent $i$ may only improve the outcome for all other agents\footnote{Any such deviation would make everyone prefer $\dictator(\ellsmi,\rli)$ over $\po$.}. I.e., non truthful mechanisms do not necessarily lead to bad outcomes in facility location games and thus should not be excluded from consideration. Furthermore, while the literature on strategic facility location is primarily concerned with truthful mechanisms, other areas of mechanism design have encountered a fair number of non truthful mechanisms, e.g., first-price or position auctions in revenue maximization settings. For such mechanisms, a standard approach is to analyse the equilibria of the ensuing game and establish good Price of Anarchy (PoA) bounds (e.g.,~\cite{CaragiannisKKKLLT15}).

\subsection{Our Contributions}
While the consensus mechanism suggests that positive results are possible in ML augmented framework, it does not provide a satisfactory solution for the most interesting case, when the number of agents $n$ is large. Indeed, its robustness guarantee of $\Omega(n)$ scales linearly with $n$ (e.g., when one agent $i$ is located exactly at the prediction $\elli=\po$ and all other agents stay at the same far away location $\elli[j]=P\in\metric$ for all $j\ne i$) and it does not allow for any trade-offs between consistency and robustness guarantees. 

In this paper, we propose a novel mechanism, which we term \harm \ and which to the best of our knowledge has not been considered before. The mechanism chooses a reported location $\rli\in\{\rli[1],\ldots,\rli[n]\}$ with a probability \emph{inversely} proportional to its distance $d(\rli,\po)$ to the predicted location $\po$ plus a constant $\pam$, i.e., $\Prx{f \gets \rli}$ is proportional to $1/(d(\po,\rli)+\pam)$. The intuition for the inversely proportional 
selection rule is to penalize agents for reporting remote 
locations to be more robust against manipulations by adversarial agents\footnote{Recall a major issue of $\dictator(\ells)$: a single adversarial agent can make everyone's expected cost arbitrarily large. With an inversely proportional rule, reporting a remote location has a bounded impact on the expected cost of any agent.}. The parameter $\pam$ allows us to balance consistency and robustness guarantees: smaller $\pam$ means better consistency and makes \harm \ resemble the consensus mechanism, while larger $\pam$ brings \harm \ closer to the $\dictator$ mechanism.   
The hope is that if the predicted optimal location $\po$ is close to accurate, then many agents should benefit by shifting their reports closer to $\po$ (agent $i$ gets a higher chance of being selected at the reported location $\rli$ at the cost of higher distance to their true location $d(\elli,\rli)$). On the other hand,
when prediction $\po$ is highly inaccurate, the hope is that only a few agents decide to move closer to $\po$ and we get a good robustness guarantee. This is also the situation where the constant parameter $\pam$ helps to reduce the impact of a single agent on the mechanism's outcome. 

\smallskip
\noindent \textbf{Results.} 
First, while $\harm$ is not strategy-proof, we show that it significantly limits the agents' undominated action space --- in a good direction. Specifically, the best response of each agent $i$ is to report her preferred location, $\elli$ or the predicted location, $\po$, or any point on the shortest path between $\elli$ and $\po$, when she is indifferent between $\elli$ and $\po$. In other words, each agent can only choose from the shortest paths between her preferred location and the predicted location. We then prove that Pure Nash Equilibrium (PNE) always exists via Kakutani's
fixed-point theorem (Section~\ref{sec:pure_vs_mixed}). The only assumption we use is that every pair of points in the metric space is connected by a continuous shortest path. This is a rather mild assumption, as one can easily embed a given discrete metric into a continuous one by taking a convex combination over the points in the space. In particular, for any discrete metric $\metric$, one can simply extend the metric to $\overline{\metric}$ by adding a continuous shortest path between each pair of discrete points.

Second, we analyse the Price of Anarchy (PoA) of the ensuing game induced by the $\harm(\pam)$ mechanism. We assume that the trade-off parameter $\pam=c\cdot\pac$ is related to the average distance $\pac\eqdef\frac{1}{n}\sum_{i\in[n]}d(\po,\elli)$ to the prediction $\po$ for a constant $c$. 
When predicted location $\po$ is $\gamma$-approximation to the optimum, $\harm$ is $\gamma(1 + 2c)$-consistent (Theorem~\ref{thm:consistency_poa}). 
On the robustness side, we show that $\harm$ is $O(1 + 1/c^3)$-robust (Theorem~\ref{thm:harm-robust-metricspace}) in general, and give a better robustness guarantee of $O(1 + 1/c^2)$, if the metric is strictly convex (Theorem~\ref{thm:harm-robust-normspace}). I.e., for a small constant $c=\eps/2$ we achieve 
$(1+\eps)$-consistency that smoothly degrades with the prediction's accuracy, and respectively $O(1+1/\eps^3)$ and $O(1+1/\eps^2)$ robustness\footnote{Note that our consistency/robustness guarantees rely on the trade-off parameter $\pam$. I.e., we rely on the approximately correct estimation of the average distance $\pac$, which is arguably not a very strong assumption on the setting ($\pac$ is a single number with simple dependencies on each agent's individual location $\elli$).
Also note that, if $\pam$ is underestimated ($\frac{\pam}{\pac}=o(\eps)$), then we get better consistency and worse robustness; and if $\pam$ is overestimated (
$\frac{\pam}{\pac}=\omega(\eps)$), then we get better robustness.}. Furthermore, if $\pam = 0$, we obtain $O(n)$-robustness (Theorem~\ref{thm:constant_agents}) and $1$-consistency, which translates into a constant robustness and $1$-consistency for a constant $n=O(1)$ number of agents (see the discussion in the full version).

\subsection{Other Related Work}
\label{sec:related}

Mechanism design without payments often falls short of producing the optimal solution, leading to potentially suboptimal results. Agrawal et al.~\cite{agrawal2022learning} initiated the application of machine learning-augmented algorithms in this field to overcome these limitations. 
 For the utilitarian objective of minimizing the average distance of agents to the facility, their algorithm achieves $\sqrt{ 2c^2 + 2} / (1 +c)$-consistency and $\sqrt{ 2c^2 + 2} / (1  - c)$-robustness in the two-dimensional space, where $c \in [0, 1)$ is a parameter to be chosen by the mechanism. 
 As mentioned before, a $\sqrt{2}$-approximate strategy-proof mechanism is known and it is the best possible without predictions \cite{goel2023optimality}.  
 For the egalitarian social cost of minimizing the maximum distance to the facility, Agrawal et al.~\cite{agrawal2022learning} presented an algorithm with 1-consistency and $1 + \sqrt{2}$-robustness. For this setting, a 2-approximate  strategy-proof mechanism is known without predictions \cite{AlonFPT10,goel2023optimality}.  A very recent work \cite{rand-facility-game} studies randomized algorithms in one-dimensional and two-dimensional settings. For the obnoxious version of maximizing the total distance to the obnoxious facility with predictions, Istrate and Bonchis recently considered various metrics such as circles, trees, and one-dimensional and two-dimensional hypercubes \cite{istrate2022mechanism}. 

Very recently, an ML augmented mechanism was developed for the $k$-facility location game in  general metric spaces \cite{barak2024macadvicefacilitylocation}. 
However, it assumes predictions for each point, requires assigning points  to facilities in a balanced manner, and does not provide an absolute robustness guarantee.

Facility location problems have been studied (without predictions) for various metrics. See \cite{AlonFPT10} for general metrics, \cite{AlonFPT10,meir2019strategyproof} for circles, \cite{AlonFPT10,FeldmanW13} for trees, and \cite{meir2019strategyproof,walsh2020strategy,goel2023optimality,el2023strategyproofness} for $d$-dimensional Euclidean spaces. For characterization of the
space of startgy-proof mechanisms, see \cite{moulin1980strategy} and \cite{peters1993range}, which consider the one-dimensional and two-dimensional spaces, respectively. 

For the obnoxious version of maximizing the total distance to the obnoxious facility with predictions, Istrate and Bonchis recently considered various metrics such as circles, trees, and one-dimensional and two-dimensional hypercubes \cite{istrate2022mechanism}.

 Machine learning augmented algorithms have found other applications in mechanism design without payments. For the unrelated machine makespan minimization problem, Balkanski \etal gave a $O(1)$-consistent and $O(n)$-robust algorithm, which asymptotically matches the best possible worst-case bound \cite{nisan1999algorithmic,0001KK23}. Gkatzelis et. al.~\cite{GkatzelisKST22} revisited  scheduling games
and network formation games, demonstrating that leveraging predictions can improve the price of anarchy.

In parallel to Agrawal et al. \cite{agrawal2022learning}, 
Xu and Lu \cite{XuL22} studied several mechanism design problems, some of which involve payment (single item auction and frugal path auction) and some not (truthful job scheduling and two facility game on a line). Notably, they gave a mechanism that is $(1 + n/2)$-consistent and $(2n-1)$-robust for the two facility location game on a line.

Furthermore, ML augmented algorithms have found broader applications in online algorithms. Given the extensive literature on this topic, we only provide a few sample examples: paging \cite{LykourisV21},  data structures \cite{kraska2018case,mitzenmacher2018model}, knapsack \cite{ImKQP21}, load balancing \cite{lattanzi2020online}, scheduling \cite{NEURIPS2018_73a427ba}, rent-or-buy \cite{NEURIPS2018_73a427ba}, graph problems \cite{azar2022online}, covering problems \cite{bamas2020primal}, facility location problems with no private information \cite{fotakis2021learning,jiang2021online}, online learning \cite{bhaskara2020online}, and fair allocations \cite{BanerjeeGGJ22,BanerjeeGHJM023}. They have also been effectively used to speed up offline algorithms, e.g. min-cost bipartite matching \cite{DinitzILMV21} and shortest paths \cite{lattanzi2023speeding,feijen2021dijkstras}.

\subsection{Roadmap}

In Section~\ref{sec:prelim}, we define the problem and our $\harm$ mechanism. In Section~\ref{sec:strategies}, we discuss the undominated strategies and equilibria concepts. In Section~\ref{sec:PoA} we give analysis of our mechanism, providing consistency and robustness bounds. In Section~\ref{sec:conclusion} we conclude our paper. Due to the space constraints, most proofs are deferred to the full version of this paper.

\section{Preliminaries}
\label{sec:prelim}
In a (single) facility location problem the goal is to place a facility $f$ in a given metric space $\metric$ with a distance function $d:\metric\times\metric\to\R$ to serve a set of $n$ agents. Each agent $i\in[n]$ has a most preferred location $\elli\in\metric$ and incurs a cost of $d(f,\elli)$, when the facility is placed at $f\in\metric$.  The objective is to minimize the social cost $\SC(f,\ells)\eqdef\sum_{i\in[n]}d(f,\ell_i)$. The optimal cost is achieved by facility $o$, i.e., $\opt(\ells)=\min_{o' \in \metric} \sum_{i\in[n]} d(o', \elli)=\sum_{i\in[n]} d(o, \elli)$. The facility location $f$ is an $\alpha$-approximation if $\SC(f,\ells)\le\alpha\cdot\opt(\ells)$. 

In the strategic setting, the location $\elli$ is private information of each agent $i\in[n]$, i.e., $i$ may report a different location $\rli$ to the mechanism. The mechanism $f: \metric^n \to \metric$ takes the reported location vector $\rls = (\rli[1], \rli[2], \ldots, \rli[n])$ and outputs the facility $f(\rls)$. For notational convenience, we may let $f$ denote the mechanism's outcome. If the mechanism is randomized, then the outcome is a distribution $\pi(\metric)$ of the facilities in $\metric$. 
We assume that agents are selfish and rational, i.e., each agent $i$ tries to minimize their own expected cost, $\costi\eqdef\Ex{d(f, \elli)}$ for the chosen 
facility $f(\rlsmi,\rli)$, when $i$ reports $\rli$ and the remaining $n-1$ agents  report $\rlsmi$. 
The mechanism is called {\em dominant strategy incentive compatible} (DSIC) if $\Ex{d(f(\rli,  \ellsmi), \elli)} \ge \Ex{d(f(\ells), \elli)}$ for all $i\in[n]$ and all $\ells\in\metric^n$. E.g., a random dictatorship mechanism that selects an agent $i$ uniformly at random and lets $i$ to place the facility at their preferred location $\elli$, i.e., $f: \metric^n\to\pi(\metric)$ is given by $f\sim\textsc{uni}\{\elli[1],\ldots,\elli[n]\}$, is a DSIC mechanism \cite{alon2009strategyproof}. On the other hand, a non-truthful mechanism induces a game among $n$ agents. Similar to DSIC mechanisms, we assume that agents only employ undominated strategies\footnote{No agent $i$ plays a strategy $\rli$ with a higher expected cost than another strategy $\rli'$ on every profile $\ellsmi\in\metric^{n-1}$} (see  
Claim~\ref{claim:about_x_first_claim}).
We further assume that agents 
reach Nash Equilibrium $\rls\in\texttt{NE}(\ells)$, i.e., at a given reported profile $\rls=(\rli[1],\ldots,\rli[n])$ no agent $i$ attains smaller expected cost by a unilateral deviation $(\rli',\rlsmi)$. The Price of Anarchy (PoA) of the mechanism $f:\metric^n\to\pi(\metric)$ is the worst-case ratio of the expected social cost attained at equilibrium and the optimal social cost.
\be
\label{eq:PoA_def}
\texttt{PoA}(f)\eqdef\max\limits_{\ells\in\metric^n}\max\limits_{\rls\in\texttt{NE}(\ells)}\frac{\Ex{\SC(f(\rls),\ells)}}{\opt(\ells)}
\ee

\paragraph{ML Augmented Mechanism Design for Facility Location Problem.} The mechanism designer in addition to the reports $\rls$ is given a prediction $\po$ of the optimum location $o$. Formally, the  mechanism is $f(\rls,\po):\metric^{n+1}\to\pi(\metric)$ (to emphasize the value of $\pam$ used, we may put it as a subscript; we have $f_{\pam}(\cdot,\cdot)$ in our case, where $\pam=\frac{c}{n}\cdot\SC(\po,\ells)$). Consistency and robustness are two standard measures that describe the performance of an algorithm with predictions \cite{LykourisV21}. 

The {\em consistency} captures how well the algorithm performs, 
when prediction is correct ($\po=o(\ells)$): an algorithm (or a truthful mechanism) is $\alpha$-consistent if it achieves an $\alpha$-approximation to the optimum.
In case of non truthful mechanisms, we say that a mechanism $f(\rls,\po)$ is $\alpha$-consistent if its price of anarchy \eqref{eq:PoA_def} is at most $\alpha$ on any instance $\ells\in\metric^n$ when $\po=o(\ells)$.
It is often unreasonable to expect that the predicted solution $\po$ is $100\%$ accurate. Thus it is important to understand how approximation guarantee degrades with the quality of the prediction. To this end, we measure the quality of the prediction $\po$ as the ratio of $\SC(\po,\ells)$ and $\opt(\ells)$.\footnote{One can also use $d(\po,o)$ as a measure of the accuracy of the prediction. A drawback of this approach is that $o$ might not be unique, or that a pair of far away points may yield similar social cost.} Namely, we say that prediction is {\em $\gamma$-accurate} if $\frac{\SC(\po,\ells)}{\opt(\ells)}\le\gamma$.

In contrast, {\em robustness} captures how much of the worst-case guarantees the mechanism retains, when the prediction is arbitrarily wrong. Namely, mechanism $f(\rls,\po)$ is $\beta$-robust, if the price of anarchy of the game induced by $f$ is at most $\beta$, i.e., 
\[
\beta\ge\max\limits_{\ells,\po}\max\limits_{\rls\in\texttt{NE}(\ells,f_{\pam})}\frac{\Ex{\SC(f_{\pam}(\rls,\po),\ells)}}{\opt(\ells)}.
\]

The ML augmented framework studies which $\alpha$-consistency and $\beta$-robustness guarantees are attainable. Typically, there is a trade off between feasible $\alpha$ and $\beta$. We provide guarantees on the robustness and consistency as functions of the prediction accuracy $\gamma$ and the ratio $c=\frac{\pam}{\pac}$ between given parameter $\pam$ and average cost $\pac\eqdef\frac{1}{n}\SC(\po,\ells)$ for the predicted location $\po$.

\vspace{-2mm}
\paragraph{Harmonic Mechanism.} We analyze the following \harm \ mechanism, which installs facility $f$ at $\rli$ with probability inversely proportional to $d(\rli, \po) + \pam$ for a constant $\pam$.
\begin{algorithm}
\caption{Harmonic mechanism $\harm(\pam)$}\label{alg:harmonic}
\KwData{Reported locations $\rls=(\rli[1],\ldots,\rli[n])$, prediction $\po\in\metric$}
\KwResult{Facility $f\sim\pi\{\rli[1],\ldots,\rli[n]\}$}
\lFor{$i\in[n]$}{
let $d_i= d(\rli, \po)$
}
Choose $\Prx{f\gets \rli}\eqdef\frac{1 / (d_i + \pam)}{\sum_{j=1}^n 1/(d_j + \pam)}$ for $i\in[n]$ 
\tcp*{$\Prx{f\gets \rli}\propto\frac{1}{d_i + \pam}$ (proportional to)}
\end{algorithm}

This mechanism is not dominant strategy incentive compatible (DSIC) in general. However, as we will show shortly, the set of undominated strategies for each agent is rather limited. In other words, \harm \ mechanism is not too far from strategy-proof mechanisms. We further analyse how far the social cost at a Nash equilibrium is from the optimum social cost, namely the price of anarchy under \harm \ selection rule $f$.

\section{Strategies and Equilibria in Harmonic Mechanism}
\label{sec:strategies}
We first describe the set of strictly dominant strategies of each
agent $i\in[n]$ under \harm\ mechanism. For notational brevity, 
we use $\elli[j]$ for $j \neq i$ in place of $\rli[j]$, as we 
consider a fixed agent $i$'s strategy for all possible reports of the other agents and one can simply pretend that $\locsmi$ is a variable vector to agent $i$. 

\subsection{Undominated strategies}
\label{sec:undominated}
We consider the expected $\cost_i(\rli)$ of a fixed agent $i$ as a function of her report $\rli$, while assuming that the reports $\locsmi$ of other agents follow a distribution\footnote{The main focus of our paper is on pure Nash equilibria, but here we allow randomization in other agents' strategies.} $\distsmi$. Then 
\be
\label{eq:cost-of-i}
\cost_i(\rli) = \Exlong[\locsmi\sim\distsmi]{
\frac{\sum_{j\ne i}\frac{d(\ell_i,\ell_j)}{d(\ell_j,\po)+\pam}+\frac{d(\ell_i,\rli)}{d(\rli,\po)+\pam}}{\sum_{j\ne i}\frac{1}{d(\ell_j,\po)+\pam}+\frac{1}{d(\rli,\po)+\pam}}
}=
\Exlong[\locsmi]{
\frac{C_1(\locsmi)
+\frac{x}{y+\pam}}{C_2(\locsmi) +\frac{1}{y+\pam}}
},
\ee
where agent $i$ can control only two parameters $x\eqdef d(\ell_i,\rli)\ge 0$ and $y\eqdef d(\rli,\po)\ge 0$ by possibly misreporting her true location, while the remaining terms $C_1\ge 0$ and $C_2> 0$ in the numerator and denominator respectively only depend on the reports of other agents $\locsmi$. 
Thus, we can study agent $i$'s cost as a function of only $x,y\in\R_+$, and each fixed realization of $\locsmi\sim\distsmi$. For convenience we continue to use the same 
notations $C_1$, $C_2,$ and $\cost_i$ for each fixed realization of $\locsmi$:~~
$
\cost_i(x,y,\locsmi)
=
\frac{C_1+x/(y+\pam)}{C_2+1/(y+\pam)}.
$
A few simple observations are in order.
\begin{observation}
\label{obs:monotone_x}
 The cost $\cost_i(x,y,\locsmi)$ is a strictly increasing in $x$ for any $C_1,C_2,y,\pam\ge 0$.    
\end{observation}
By triangle inequality $x+y=d(\elli,\rli)+d(\rli,\po)\ge d(\elli,\po)$ and $x+d(\elli,\po)\ge y$. Hence,
\begin{observation}
\label{obs:x_abs}
 $x\ge \abs{d(\elli,\po)-y}$.
\end{observation}
We further note that player $i$ would never choose $y$ larger than $d(\elli,\po)$.

\begin{restatable}{claim}{claimaboutx}
\label{claim:about_x_first_claim}
The strategy of agent $i$ with $y>d(\elli,\po)$ is strictly dominated by the strategy $x=0$, $y=d(\elli,\po)$ for any $C_1,C_2\ge 0$. 
\end{restatable}

\begin{proof}
    
Let us denote $d(\elli,\po)$ by $d_i$. By observation~\ref{obs:x_abs} we have $x\ge y-d_i\eqdef x_0$, and by observation~\ref{obs:monotone_x} the $\cost_i(x,y,\locsmi)$ is at least $\cost_i(x_0,y,\locsmi)$, where $y=x_0+d_i$. We need to show that 
\[
\cost_i(x_0,y)=\frac{C_1+\frac{x_0}{x_0+d_i+\pam}}{C_2+\frac{1}{x_0+d_i+\pam}}>\cost_i(0,d_i)=\frac{C_1}{C_2+\frac{1}{d_i+\pam}}.
\]
This inequality holds, since the numerator $C_1+\frac{x_0}{x_0+d_i+\pam}$ is larger than the numerator $C_1$, and the denominator $C_2+\frac{1}{x_0+d_i+\pam}$ is smaller than the denominator
$C_2+\frac{1}{d_i+\pam}$.
\end{proof}

Therefore, $d(\rli,\po)=y\in[0,d(\elli,\po)]$. By observation~\ref{obs:monotone_x}, agent $i$ chooses minimal possible $x$ for any given $y$ and any $\locsmi$. Now, $x\ge d(\elli,\po)-y$ by observation~\ref{obs:x_abs}. If $\metric$ is a {\em continuous space}\footnote{Formally, for any two points $P_1,P_2\in\metric$ and $x\in[0,d(P_1,P_2)]$ there exists a point $P\in\metric$ such that $d(P,P_1)=x$ and $d(P,P_2)=d(P_1,P_2)-x$}, then there always exists a $\rli\in\metric$ with $x\eqdef d(\rli,\elli)= d(\elli,\po)-y$. I.e., player $i$'s undominated strategies in the continuous metric space are $\rli$ on a shortest path $\shortpath(\po,\elli)$ between $\po$ and $\elli$.
\begin{claim}
\label{cl:shortest_path}
If $\metric$ is a continuous metric space, then reports on the shortest paths $\rli\in\shortpath(\po,\ell_i)$ strictly dominate other strategies of agent $i$.
\end{claim}

\subsection{Equilibria Concepts}
\label{sec:pure_vs_mixed}
Many solution concepts may be used to describe the outcome of a game, such as Pure Nash (PNE), Mixed Nash (MNE), Correlated (CE), and Coarse Correlated (CCE) equilibria. In case of the game induced by the \harm \ mechanism $f$, the solution concepts that involve randomization are not ideal. Indeed, an agent $i$ might find it difficult to compute their expected cost \eqref{eq:cost-of-i} for a fixed reported location $\rli$ let alone finding the best response over all possible $\rli\in\shortpath(\po,\elli)$. E.g., if $n-1=10$ agents randomize in $\locsmi$ between two pure strategies in a mixed Nash Equilibrium, then \eqref{eq:cost-of-i} would already have more than $1000$ fraction terms. 

Thus we adopt the Pure Nash Equilibrium (PNE) solution concept to describe the outcome of \harm \ mechanism $f$. It is important to keep in mind that, unlike MNE, PNE may not exist. We show below that PNE always exists in a continuous metric space $\metric$, i.e., in a metric where every pair of points has a continuous shortest path between them. The proof refers to the full version.
\begin{restatable}{theorem}{PNEexistence}
    \label{thm:Pure_Nash_existence}
    Let $\metric$ be continuous (any pair of points has a shortest path), then for any initial positions $\locs=(\ell_i)_{i\in[n]}$ of $[n]$ agents there is a Pure Nash Equilibrium under \harm\ mechanism. 
\end{restatable}

\begin{proof}

Let us fix a single shortest path in $\shortpath(\po,\ell_i)$ parameterized as $L_i=[0,1]$ (by the ratio $d(\rli,\po)/d(\ell_i,\po)\in[0,1]$ for $\rli\in\shortpath(\po,\ell_i)$) per each agent $i\in[n]$. We shall apply Kakutani fixed-point theorem for the simultaneous best responses of all $n$ players, in order to prove the existence of PNE. 

Specifically, for the convex and compact subset $X\eqdef\prod_{i=1}^n L_i=[0,1]^n$ of Euclidean space, we consider the best response $\br: X\to 2^{X}$ function which is defined as follows. Any $\vect{x}\in X$ corresponds to a location profile $\locs\in\metric^n$ on the specific shortest paths from each $\ell_i$ to the predicted location $\po$; for every player $i$ we consider all her potential best responses $\bri(\locsmi)$ with respect to the locations $\locsmi$ of other players and restricted to her fixed shortest path $L_i$; we define $\br(\vect{x})\eqdef \left(\bri(\locsmi)\right)_{i=1}^{n}$.  

By Claim~\ref{cl:pure_equilibrium} (we shall formally state and prove it shortly in the next subsection) each $\bri(\locsmi)$ is either a single point $\ell_i$ ($x_i=1$,) or $\po$ $(x_i=0)$, or it is the whole line segment $L_i$. Hence, $\br(\vect{x})$ is a convex set for every $\vect{x}\in X$. It is also not hard to verify that $\br$ has a closed graph in the product topology on $X\times X$, i.e., the set $\{(\vect{x},\vect{y}) | \vect{x}\in X, \vect{y}\in X, \vect{y}\in\br(\vect{x})\}$ is closed. The Kakutani fixed-point theorem states that any such function must have a fixed point, i.e., $\exists ~\vect{x}^*\in X$ such that $\vect{x}^*\in \br(\vect{x}^*)$. This $\vect{x}^*$ corresponds to the Pure Nash Equilibrium of \harm\ mechanism.
\end{proof}

For discrete (non continuous) metric spaces $\metric$ to allow randomness, we can let agents to explicitly report their locations as a distribution over a finite number of points in $\metric$. I.e., we consider an extension of the metric space $\metric$ to a larger metric space $\overline{\metric}$ over all finite convex combinations $\overline{\metric}=\{ \boldsymbol{\alpha}=\sum_{i\in[k]}\alpha_i\cdot\ell_i ~|~ \sum_{i\in[k]}\alpha_i=1,~\alpha_i\ge 0,~\ell_i\in\metric\}$, where a naturally induced metric in $\overline{\metric}$ is given by the earth mover's distance between two probability distributions $\boldsymbol{\alpha},\boldsymbol{\beta}$ over points in $\metric$. Then $\overline{\metric}$ is a continuous metric space and thus PNE does exist in $\overline{\metric}$.
Furthermore, when many agents participate in \harm\ mechanism, a reasonable approximation and/or simplification of their expected costs (due to good concentration) can be done by calculating expectations of the numerator and denominator in \eqref{eq:cost-of-i} separately and then using the ratio of expectations instead of 
the expected ratio. This approach corresponds exactly to agents 
playing pure strategies in $\overline{\metric}$.    

\paragraph{Best Response in Pure Nash Equilibria}
We next show in Claim~\ref{cl:pure_equilibrium} which locations on the shortest paths could be the best response of agent $i$ for a fixed profile of other agents' reports $\locsmi$: $i$ only needs to decide between two choices  $\rli\in\{\ell_i,\po\}$ and only when she is indifferent, then $i$ may also play any $\rli\in\shortpath(\po,\ell_i)$. 

\begin{restatable}{claim}{pureequlibriumequations}
\label{cl:pure_equilibrium}
Given a profile $\locsmi$ of agents $[n]\setminus \{i\}$ locations, agent $i$'s best response is to only report
    \begin{itemize}[nosep]
        \item her true location $\rli=\elli$ when         
            $\cost_i(\elli,\locsmi)<d(\elli,\po)+\pam$; 
        \item predicted location $\rli=\po$ when 
            $\cost_i(\elli,\locsmi)>  d(\elli,\po)+\pam$; then $\cost_i(\rli,\locsmi) \ge d(\elli,\po)+\pam$; 
        \item any point $\rli\in\shortpath(\po,\elli)$ when $\cost_i(\elli,\locsmi) = d(\elli,\po)+\pam$; then $\cost_i(\elli,\locsmi)= \cost_i(\rli,\locsmi)$. %
    \end{itemize}
\end{restatable}

\begin{proof}
By equation \eqref{eq:cost-of-i} $i$ chooses $\rli$ with $x,y$ that minimize
\[
\cost_i(x,y,\locsmi)=
\frac{\sum_{j\ne i}d(\elli,\ell_j)/(d(\ell_j,\po)+\pam)+x/(y+\pam)}{\sum_{j\ne i}1/(d(\ell_j,\po)+\pam)+1/(y+\pam)},
\] 
where $x=d(\elli,\rli)\ge 0$ and $d(\elli,\po)\ge y=d(\rli,\po)\ge 0$. Without loss of generality we can assume that $x+y=d(\elli,\po)$. 
We consider $\cost_i(\rli,\locsmi)-\cost_i(\elli,\locsmi)=
\cost_i(x,y,\locsmi)-\cost_i(0,x+y,\locsmi)$. To simplify notations let us denote the pairwise distances $d_{jk}\eqdef d(\ell_j,\ell_k)$ for $j,k\in[n]$, $d_j\eqdef d(\ell_j,\po)$ for $j\in[n]$, and the respective denominators
$\denomr=\sum_{j\ne i}\frac{1}{d_j+\pam}+\frac{1}{y+\pam}$ and $\denom=\sum_{j\ne i}\frac{1}{d_j+\pam}+\frac{1}{x+y+\pam}$. Then
\begin{align}
    \cost_i(\rli,\locsmi) - \cost_i(\ell_i,\locsmi) &=
    \frac{x}{(x+y+\pam)\cdot \denomr}-
    \left(\sum_{j\ne i}\frac{d_{ij}}{d_j+\pam}\right)\cdot
    \frac{1/(y+\pam)-1/(x+y+\pam)}{\denom\cdot\denomr} \nonumber
 \\
    &=\frac{x}{(y+\pam)\cdot \denomr}-
    \frac{1}{\denom}\left(\sum_{j\ne i}\frac{d_{ij}}{d_j+\pam}\right)\cdot
    \frac{x}{(y+\pam)\cdot(x+y+\pam)\cdot\denomr} \nonumber
\\
    &= \frac{x}{(y+\pam)\cdot(x+y+\pam)\cdot\denomr}\cdot
    \Big(x+y+\pam-\cost_i(\ell_i,\locsmi)\Big). \label{eq:cost_difference}
\end{align}
When $x+y+\pam-\cost_i(\ell_i,\locsmi)=d(\ell_i,\po)+\pam-\cost_i(\ell_i,\locsmi)>0$, the expression \eqref{eq:cost_difference} is minimized only for $\rli=\elli$ with $x=0,y=d(\elli,\po)$. When $x+y+\pam-\cost_i(\elli,\locsmi)<0$, expression \eqref{eq:cost_difference} is minimized only for $\rli=\po$ with $y=0,x=d(\ell_i,\po)$. When $x+y+\pam-\cost_i(\elli,\locsmi)=0$ we get $\cost_i(\rli,\locsmi) - \cost_i(\elli,\locsmi) =0$ for any $x,y: x+y=d(\elli,\po)$.
This almost concludes the proof, as we get all required bounds on $\cost_i(\elli,\ellsmi)$. 

The only missing part is to show that $\cost_i(\rli,\ellsmi)\ge d(\elli,\po)+\pam$, when 
$\cost_i(\elli,\ellsmi)> d(\elli,\po)+\pam$ and $\rli=\po$. Note that in this case $y=d(\rli,\po)=0$ and $x=d(\rli,\elli)=d(\elli,\po)$. Then $\cost_i(\elli,\ellsmi)=\cost_i(0,x)=\frac{C_1}{C_2+1/(x+\pam)}> x+\pam$, i.e., $C_1\ge C_2(x+\pam)+1$, while $\cost_i(\po,\ellsmi)=\cost_i(x,0)=\frac{C_1+x/\pam}{C_2+1/\pam}$. 
We get the desired bound on $\cost_i(\po,\ellsmi)$, when using the lower bound on $C_1\ge C_2(x+\pam)+1$: $\cost_i(\po,\ellsmi)=\frac{C_1+x/\pam}{C_2+1/\pam}\ge
\frac{C_2(x+\pam)+1+x/\pam}{C_2+1/\pam}=x+\pam$.
\end{proof}

\section{Price of Anarchy}
\label{sec:PoA}
Here we study the Price of Anarchy (PoA) of the PNE of \harm \ mechanism $f$ in metric $\metric$ or $\overline{\metric}$. We will make extensive use of Claims~\ref{cl:shortest_path},~\ref{cl:pure_equilibrium} to prove PoA guarantees for the consistency and robustness of \harm \ mechanism.
We derive two different PoA bounds for robustness: (a) when $\metric$ is a strictly convex space (b) when $\metric$ is a general metric space.

We first introduce some simplified notations that would help us to use Claim~\ref{cl:pure_equilibrium}. Let $\rls=(\rli)_{i\in[n]}$ be a PNE. 
Let $\tdi\eqdef d(\po,\ell_i)$ denote the distance between predicted location $\po$ and the true location $\elli$ of agent $i\in[n]$. In equilibrium, some of the agents may report different locations $\rli$, we denote by $d_i\eqdef d(\po,\rli)$ the distance from $\po$ to the reported location of $i\in[n]$. 
Note that player $i$ will always have $t_i\ge d_i$ in an equilibrium. To write down the  $\cost_i(\rli)$ of agent $i$ we will use $c_{ij}\eqdef d(\ell_i,\rli[j])$ (note that $c_{ij}\ne c_{ji}$) that denotes the distance between true location of $i\in[n]$ and the reported location of $j\in[n]$. In particular, $c_{ii}=d(\ell_i,\rli)=t_i-d_i$. Then 
\be
\label{eq:cost_short}
\cost_i(\rli)=\frac{\sum_{j\in[n]}c_{ij}/(d_j+\pam)}{\sum_{j\in[n]}1/(d_j+\pam)}
\ee

By Claim~\ref{cl:pure_equilibrium} each agent either reports $\rli=\po$, or $\rli=\ell_i$, or $\rli\in\shortpath(\po,\ell_i)$ in a PNE. Let 
\[S\eqdef\{i\in [n] ~|~\rli=\po\},\quad
T\eqdef\{i\in [n] ~|~\rli=\ell_i\}, \quad U\eqdef [n]\setminus (S\cup T)\] and also use $\overline{S}=[n]\setminus S= U\cup T$. We also use $D\eqdef\sum_{j\in[n]}1/(d_j+\pam)$ to denote the denominator of each $\cost_i(\rls)$ for $i\in[n]$.

\subsection{Consistency}
\label{sec:consistency}

In this section, we show that the $\harm(\pam)$ mechanism can achieve a solution arbitrarily close to the optimum depending on the quality of the prediction. Specifically, when the predicted location $\po$ gives an almost optimum estimation of the true optimum ($\sum_{i \in [n]} d(\ell_i, \po) \approx \sum_{i \in [n]} d(\ell_i, o)$) and $\frac{\pam}{\opt/n}=O(\eps)$, we are able to show that PoA of $\harm(\pam)$ mechanism is close to $1$. 
\begin{restatable}{theorem}{consistencypoa}
\label{thm:consistency_poa}    
    If $\po$ is $\gamma$-accurate, i.e., $\frac{\SC(\po, \ells)}{\opt(\ells)} \le \gamma$, and $c \eqdef \pam / \pac$ for an average cost $\pac=\frac{1}{n}\SC(\po,\ells)$ of the prediction $\po$, then $\harm(\pam)$ mechanism is $ \gamma\cdot (1 + 2c)$-consistent.
\end{restatable}

\begin{proof}
    Let $\rls$ be a PNE reported by the agents. Then the expected social cost of the \harm \ mechanism $f$ is as follows, 
\[
   \SC\eqdef \sum_{i \in [n]} \cost_i(\rli) = \sum_{j \in [n]} \frac{1 / (d_j +\pam) }{\sum_{j \in \overline{S}} 1 / (d_j + \pam) + |S| / \pam} \sum_{i \in [n]} c_{ij}
\]
We denote the denominator of the equation above by $D \eqdef\sum_{j \in \overline{S}} 1 / (d_j + \pam) + |S| / \pam$. Claim~\ref{cl:pure_equilibrium} describes three groups of agents $S, U,$ and $T$ in a PNE.
We rewrite $\SC$ by partitioning agents into $S$ and $\overline{S} = U \cup T$ as follows,
\[
    \SC=\frac{|S| / \pam}{D} \sum_{i \in [n]} t_i + \sum_{j \in \overline{S}} \frac{1 / (d_j + \pam)}{D} \sum_{i \in [n]} c_{ij}. 
\]
In the equilibrium, we have $t_i + \pam \ge \cost_i(\elli, \rlsmi) = \cost_i(\rli)$ for $i \in \overline{S}$.
We add these inequalities over $i \in \overline{S}$,
\begin{multline}
    \sum_{i \in \overline{S}} (t_i + \pam) \ge \frac{1}{D} \cdot \left( \sum_{i \in \overline{S}} \sum_{j \in [n]} \frac{c_{ij}}{d_j + \pam} \right)  
    = \frac{1}{D} \cdot \left( \sum_{i \in \overline{S}} \sum_{j \in \overline{S}} \frac{c_{ij}}{d_j + \pam} + \sum_{i \in \overline{S}} \sum_{j \in S} \frac{c_{ij}}{d_j + \pam} \right) \\
    = \frac{1}{D} \cdot \left( \sum_{i \in \overline{S}} \sum_{j \in \overline{S}} \frac{c_{ij}}{d_j + \pam} + \frac{|S|}{\pam}\cdot \sum_{i \in \overline{S}} t_i \right), \label{eqn:consist-1}
\end{multline}
where we used that $c_{ij} = t_i$ and $d_j = 0$ for any $j \in S$ and $i \in \overline{S}$ to get the second equality. We rewrite Right Hand Side (RHS) of  ~\eqref{eqn:consist-1} using $\SC$ as follows
\begin{multline*}
    \textsf{RHS}\eqref{eqn:consist-1} = \SC - \frac{1}{D}\cdot \left( \sum_{i \in S}\sum_{j \in \overline{S}}\frac{c_{ij}}{d_j + \pam} + \frac{|S|}{\pam} \sum_{i \in S} t_i \right) \overset{[c_{ij} \leq t_i + d_j]}{\geq} \SC - \frac{1}{D} \cdot \left( \sum_{i \in S} \sum_{j \in \bar S} \frac{t_i + d_j}{ d_j + \pam} + \frac{|S|}{\pam} \sum_{i \in S} t_i \right) \\
    \overset{[\text{def. of }D]}{=} \SC - \frac{1}{D} \cdot \left( D\cdot \sum_{i \in S} t_i  + \sum_{i \in S}\sum_{j \in \overline{S}} 
    \frac{d_j}{d_j + \pam} \right) 
    \overset{\substack{D \ge |S| / \pam,\\ \ d_j / (d_j + \pam) \le 1}}{\geq} \SC - \sum_{i \in S} t_i - \frac{1}{|S| / \pam} \sum_{i\in S}\sum_{j \in \overline{S}} 1 \\
    = \SC - \sum_{i \in S} t_i - \pam\cdot |\overline{S}|
\end{multline*}
Since $\sum_{i \in \overline{S}} (t_i + \pam)\ge \textsf{RHS}\eqref{eqn:consist-1}$, we get $ \SC \leq \sum_{i \in [n]} t_i + 2\pam |\overline{S}|$. 
Recall that $\pam = \frac{c}{n} \cdot \sum_{i\in[n]}d(\elli,\po)=\frac{c}{n} \cdot \sum_{i\in[n]}t_i$. Then, $\SC \leq (1 + 2c |\overline{S}| / n) \sum_{i \in [n]} t_i \leq (1 + 2c) \sum_{i \in [n]} t_i = (1 + 2c) \cdot \SC(\po, \ells)$. Using the fact that $\po$ is $\gamma$-accurate, $(1 + 2c) \SC(\po, \ells) \le \gamma (1 + 2c) \cdot \opt(\ells)$.
\end{proof}

\subsection{Robustness}
\label{sec:robust}

Here, we show robustness guarantees for arbitrary bad predicted location $\po$ for the \harm \ mechanism. Our guarantees  depend on $c=\pam/\pac$ (the larger $c$, the better) that equals to the ratio 
between the parameter $\pam$ of the \harm \ mechanism and $\pac=\frac{1}{n}\SC(\po,\ells)$. We 
will first derive a few useful lemmas and then prove two PoA guarantees: first, for the case when $\metric$ is strictly convex space, and, second, for the general metric.

Recall that \harm \ mechanism chooses each location $\rli$ for $i\in[n]$ with probability $\frac{1 / (d_i + \pam)}{D}$ where $D \eqdef \sum_{i \in [n]} 1 / (d_i + \pam)$. The expected social cost 
$\SC=\sum_{j \in [n]} \frac{1 / (d_j + \pam)}{D} \sum_{i \in [n]} c_{ij}$, which we partition into three terms by dividing $j\in[n]$ into three sets $S$, $U$, and $T$ according to the equilibrium conditions from Claim~\ref{cl:pure_equilibrium} as follows
\begin{equation}
    \label{eqn:robust-SC} 
    \SC = \frac{|S| / \pam}{D} \sum_{i \in [n]} t_i + \sum_{j \in T} \frac{1 / (t_j + \pam)}{D} \sum_{i \in [n]} d_{ij} + \sum_{j \in U} \frac{1 / (d_j + \pam)}{D} \sum_{i \in [n]} c_{ij}. 
\end{equation}
We first consider the term corresponding to the set of agents $T$, who reported their true location $\rli[j]=\elli[j]$.
It is useful to keep in mind that the random dictatorship mechanism is a $2$-approximation to the optimal social cost ~\cite{alon2009strategyproof}.
\begin{lemma}[Theorem 3.1 in \cite{alon2009strategyproof}]
    \label{thm:uniform-sampling-ratio}
    The random dictatorship mechanism (i.e., choosing location $\elli[j]$ of agent $j\in[n]$ with probability $1 / n$) is $2$-approximation.
\end{lemma}
This Lemma~\ref{thm:uniform-sampling-ratio} is a useful comparison point for estimating the contribution to the social cost by agents in $T$. 
Indeed, if we can show that the probabilities of selecting $\rli=\elli$ are at most a constant factor away from the uniform sampling, then agents in $T$ contribute no more than a constant factor compared to the optimum social cost. The next Lemma~\ref{lem:prob-TU} proves that and gives useful probability estimates for selecting agents $i\in U$ under \harm \ mechanism.
\begin{restatable}{lemma}{probTU}
    \label{lem:prob-TU}
    For all $i \in T\cup U$, the probability 
    that location $\rli$ is selected  is 
    at most
     \be 
     \label{eq:prob-TU}
     \frac{1 / (d_i + \pam)}{\sum_{j \in [n]} 1 / (d_j + \pam)} \le \frac{2}{n} \cdot \frac{t_i+\pam}{d_i+\pam}.
     \ee
\end{restatable}

\begin{proof} Inequality~\eqref{eq:prob-TU}, after simple algebraic manipulations, can be rewritten as 
\be
\label{eq:lem_prob_equivalent}
2\cdot\sum_{j\in[n]}\frac{1}{d_j+\pam}\ge\frac{n}{t_i+\pam}.
\ee
    By Claim~\ref{cl:pure_equilibrium}, we have $\cost_i(\rli)=\cost_i(\elli) \leq t_i + \pam$ for $i \in \overline{S} = T \cup U$. We multiply both sides of the latter inequality by $\frac{D}{t_i+\pam}$ and get
    \begin{multline}
    \label{eq:lem_probability}    
    \sum_{j \in [n]} \frac{1}{d_j + \pam}=\frac{D}{t_i+\pam}(t_i+\pam) \geq 
    \frac{D}{t_i+\pam}\cdot \cost_i(\rli)=
    \sum_{j \in [n]} \frac{c_{ij}}{(d_j + \pam) (t_i + \pam)}\\
    \ge\sum_{j\in[n]}\abs{\frac{1}{t_i+\pam}-\frac{1}{d_j+\pam}}\ge
    |S|\cdot\left(\frac{1}{\pam}-\frac{1}{t_i+\pam}\right)+
    \sum_{j\in\overline{S}}\left(\frac{1}{t_i+\pam}-\frac{1}{d_j+\pam}\right),
    \end{multline}
    where the second inequality holds since $c_{ij} = d(\elli, \rli[j]) \ge |d(\elli, \po) - d(\rli[j], \po)| = |t_i - d_j|$ by the triangle inequality (see Observation~\ref{obs:x_abs}); and the third inequality holds as one can separate $j\in S$ from $j\in\overline{S}$ and note that $d_j=0$ for $j\in S$.
    We get the following inequality after simple algebraic rearrangements of \eqref{eq:lem_probability} and using that $d_j=0$ for $j\in S$
    \[
    2\cdot\sum_{j\in[n]}\frac{1}{d_j+\pam}
    \ge\frac{|S|+|\overline{S}|}{t_i+\pam}+
    \frac{2|S|}{\pam}-\frac{2|S|}{t_i+\pam}
    =\frac{n}{t_i+\pam}+\frac{2|S|t_i}{\pam(t_i+\pam)},
    \]
    which obviously implies the desired inequality \eqref{eq:lem_prob_equivalent}.  
\end{proof}

The Lemma~\ref{lem:prob-TU} implies that the \harm \ mechanism chooses each location $\rlj=\elli[j]$ for agents $j \in T$ with probability at most $2/n$. I.e., the contribution of agents in $T$ to the expected social cost is at most $O(1)$ of the uniform sampling and, hence, is $O(\opt)$.
\begin{corollary}
    \label{cor:contribution-T}
    $\sum_{j \in T} \frac{1 / (t_j + \pam)}{D} \sum_{i \in [n]} d_{ij} \leq 4 \cdot \opt$.
\end{corollary}
The other two terms corresponding to agents $j\in S$ and $j\in 
U$ are the main source of \harm \ mechanism's inefficiency. 
Namely, it is impossible to achieve $O(1)$ robustness guarantee, 
when $\pam$ is significantly smaller than $\pac=\frac{1}{n}\sum_{i\in[n]}t_i$. Indeed, consider for example an instance, in which $\pam=0$, a single agent $j$ stays exactly at the predicted 
location $\elli[j]=\po$, and all other agents 
$i\in[n]\setminus\{j\}$ are all situated at the same spot in $\metric$ far away from 
$\po$. In the unique equilibrium agent $j$ reports $\rli[j]=\po$, 
while all other agents $i\ne j$ report $\elli$; the mechanism 
picks $\rlj=\po$ with large probability resulting in an $\Omega(n)$ inefficient placement of the facility. Still, when the parameter $\pam$ is not too small compared to $\pac$, we are able to show a constant approximation guarantee regardless of the predicted location $\po$.  

We begin our analysis with the term in \eqref{eqn:robust-SC} corresponding to agents $j\in S$. There we have a factor $\frac{1}{D\cdot\pam}$ and our next lemma relates this term to the parameter $c=\pam/\pac$.

\begin{restatable}{lemma}{DeltaD}
    \label{lem:DeltaD}
    Let $\pam=\frac{c}{n}\cdot\sum_{i\in[n]}d(\elli,\po)$ in $\harm(\pam)$. Then
    $\frac{1}{\pam \cdot D} \leq \max\Big\{ \frac{1/c + 1}{n}, \frac{1/c + 1}{n/c}\Big\}$.
\end{restatable}

\begin{proof}
We first get the following simple lower bound on $\pam\cdot D$.
\[
    \pam\cdot D=\sum_{i \in [n]} \frac{\pam}{\pam + d_i} = \sum_{i \in T} \frac{1}{1 + t_i / \pam} + \sum_{i \in U} \frac{1}{1 + d_i / \pam} + |S| \ge \sum_{i \in \overline{S}} \frac{1}{1 + t_i / \pam} + |S|.
\]
Next, we get a lower bound on $\sum_{i \in \overline{S}} \frac{1}{1 + t_i / \pam}$ by applying Cauchy-Schwarz inequality
\be
\label{eq:cauchy-schwarz}
    \left( \sum_{i \in \overline{S}} \frac{1}{1 + t_i / \pam} \right) \cdot \left( \sum_{i \in \overline{S}} 1 + \frac{t_i}{\pam}\right) \geq \left( \sum_{i \in \overline{S}} \sqrt{\frac{1}{1 + t_i / \pam} \cdot \left(1 + \frac{t_i}{\pam}\right)} \right)^2 = |\overline{S}|^2.
\ee
Since $\pam = \frac{c}{n} \sum_{i \in [n]} t_i$, we have $\sum_{i \in \overline{S}} 1 + t_i / \pam \le |\overline{S}|+
\sum_{i \in [n]} t_i / \pam = n / c + |\overline{S}|$. 
We plug in the last upper bound into \eqref{eq:cauchy-schwarz} to get a lower bound $\sum_{i \in \overline{S}} \frac{1}{1 + t_i / \pam} \ge \frac{|\overline{S}|^2}{n / c + |\overline{S}|}$. Hence, 
\[
\pam \cdot D \ge 
\frac{|\overline{S}|^2}{n / c + |\overline{S}|} + |S| =
\frac{|\overline{S}|\cdot (|\overline{S}|+|S|)+|S|\cdot n / c}{|\overline{S}|+n / c} =
\frac{|\overline{S}|+|S| / c}{|\overline{S}|/n + 1 / c}\ge
\frac{|\overline{S}|+|S| / c}{1 + 1 / c}.
\] 
Finally, as $|\overline{S}|/c+|S|\ge\min(n,n/c)$ for $c>0$, we get the desired bound
\[
    \frac{1}{\pam\cdot D} \le \frac{1/c + 1}{|S|/c + |\overline{S}|} \le \max\Big\{ \frac{1/c + 1}{n}, \frac{1 / c + 1}{n / c}\Big\}.
\]
\end{proof}

The Lemmas~\ref{thm:uniform-sampling-ratio}, \ref{lem:prob-TU}, \ref{lem:DeltaD} and Corollary~\ref{cor:contribution-T} are mostly enough to get a robustness guarantee of $O(1+1/c^2)$, if the set of agents $U$ in PNE $\rls$ was empty.  Unfortunately, the presence of set $U$ significantly complicates the analysis. Next, we derive two different robustness guarantees: first for an easier case, when metric $\metric$ is a strictly convex space, and later for the general metric.

\subsubsection{Robustness in Strictly Convex Spaces}
\label{sec:robust_normed}
Here, we show robustness guarantee for arbitrary bad predicted location $\po$ for the \harm \ mechanism, assuming that $\metric$ is a strictly convex space, i.e., the distance between two vectors $\vect{v_1},\vect{v_2}\in\metric$ is given by 
$d(\vect{v_1},\vect{v_2})=\norm{\vect{v_1}-\vect{v_2}}$, and $\norm{\vect{v_1}+\vect{v_2}}=\norm{\vect{v_1}}+\norm{\vect{v_2}}$ for $\vect{v_1},\vect{v_2}\ne 0$ implies that $\vect{v_1}=c\cdot\vect{v_2}$ for a $c>0$.

\begin{theorem}
    \label{thm:harm-robust-normspace}
    $\harm(\pam)$ is a $O(1+1/c^2)$-robust in strictly convex space $\metric$,
    if $\Delta=\frac{c}{n}\cdot\SC(\po,\ells)$.
\end{theorem}
To prove the theorem, we first bound bound $c_{ij}$ in $\harm$, which follows from strict convexity:

\begin{restatable}{claim}{norminequality}
    \label{cl:norm-inequality}
    For strictly convex space $\metric$, we have $c_{ij} \leq \frac{t_j - d_j}{t_j} t_i + \frac{d_j}{t_j} d_{ij}$ for $j \in U$. 
\end{restatable}

\begin{proof}
We mostly rely on the following useful property of the strictly convex spaces, whose proof is in the appendix~\ref{pf:triangle-strictly-convex-space}.
\begin{restatable}{fact}{fctnorminequality}
    \label{fact:norm-inequality}
    For arbitrary points $A,P_1,P_2\in\metric$ and a point $P\in\shortpath(P_1,P_2)$ on the shortest path
    \be
        d(A,P)\le\frac{d(P,P_2)}{d(P_1,P_2)}\cdot d(A,P_1) +
        \frac{d(P,P_1)}{d(P_1,P_2)}\cdot d(A,P_2).
    \label{eq:fact_normed}
    \ee
\end{restatable}

Recall that $c_{ij}=d(\elli,\rlj)$, $t_i=d(\elli,\po)$, and $d_{ij}=d(\elli,\elli[j])$, while $t_j=d(\elli[j],\po)$, $d_j=d(\rlj,\po)$, and $t_j-d_j=d(\rlj,\elli[j])$. We also have
$\rlj\in\shortpath(\po,\elli[j])$ and apply Fact~\ref{fact:norm-inequality} to $A=\elli, P_1=\po, P_2=\elli[j],$ and $P=\rlj$.
\end{proof}

Applying this upper bound on $c_{ij}$ for all $j\in U$ and $i\in[n]$ in the expected social cost~\eqref{eqn:robust-SC} gives
\begin{align}
    \label{eq:normed_U}
\SC &\overset{\text{Claim}~\ref{cl:norm-inequality}}{\le} 
\frac{|S|}{\pam\cdot D} \sum_{i \in [n]} t_i +
\sum_{j \in U} \frac{1 / (d_j + \pam)}{D} \sum_{i \in [n]} \Big( t_i + \frac{d_j}{t_j} d_{ij}\Big)+
\sum_{j\in T}\frac{1 / (t_j + \pam)}{D} \sum_{i \in [n]}d_{ij}\\
&=\underbrace{\left(\frac{|S|}{\pam\cdot D}+\sum_{j\in U}\frac{1}{D(d_j + \pam)}\right)\cdot\sum_{i \in [n]} t_i}_{(*)} + 
\underbrace{\sum_{j\in U}\frac{d_j/ t_j}{(d_j + \pam)D}\sum_{i \in [n]}d_{ij} +\sum_{j\in T}\frac{1/(t_j + \pam)}{D} \sum_{i \in [n]}d_{ij}}_{(**)}.\nonumber
\end{align}
We first shall get an upper bound on the second term $(**)$ in~\eqref{eq:normed_U} in a very similar way to Lemma~\ref{lem:prob-TU} and Corollary~\ref{cor:contribution-T}.
\be
    \label{eqn:star2-UB}
    (**) \le
    \sum_{j\in U}\frac{2}{n}\cdot\frac{t_j+\pam}{d_j+\pam}\cdot\frac{d_j}{t_j}\sum_{i\in[n]}d_{ij}+
    \frac{2}{n}\sum_{j\in T}\sum_{i\in[n]}d_{ij}\le
    \frac{2}{n}\sum_{\substack{j\in U\cup T,\\i\in[n]}}d_{ij}
    \le\frac{2}{n}\sum_{i,j\in[n]}d_{ij}\le 4\cdot\opt 
    ,    
\ee
where the first inequality holds by Lemma~\ref{lem:prob-TU};  the second inequality holds as $\frac{t_j+\pam}{d_j+\pam}\frac{d_j}{t_j}=\frac{1+\pam/t_j}{1+\pam/d_j}\le 1$ for all $j\in U$; the last inequality holds by Lemma~\ref{thm:uniform-sampling-ratio} similar to Corollary~\ref{cor:contribution-T}. Next, we get an upper bound on the $(*)$ term. We first note that
\begin{multline}
\label{eq:bound_star_norm}
    (*) \le
    \left(\frac{|S|}{\pam\cdot D}+\sum_{j\in U}\frac{1}{D\cdot \pam}\right)\cdot\sum_{i \in [n]} t_i =
    (|S|  + |U|)\pam \cdot\frac{1}{\pam D} \cdot \frac{\sum_{i \in [n]} t_i}{\pam}\le \\ 
    (|S|+|U|)\pam\cdot\max\left\{ \frac{1 / c + 1}{n}, \frac{1 / c + 1}{n / c}\right\}\cdot\frac{n}{c},
\end{multline}
where we got the first inequality by substituting $\frac{1}{d_j+\pam}$ with $\frac{1}{\pam}$; used 
definition of $c=\frac{\pam}{\pac}=\frac{\pam\cdot n}{\sum_{i\in[n]}t_i}$ to simplify $\frac{\sum_{i \in [n]} t_i}{\pam}$ and Lemma~\ref{lem:DeltaD} to get an upper bound on $\frac{1}{\pam D}$ in the second inequality. 
Finally, we give an upper bound on $(|S| + |U|) \pam$ in the following Claim~\ref{clm:SUDelta-norm}. The proof relies on equilibrium conditions from Claim~\ref{cl:pure_equilibrium} and   Claim~\ref{cl:norm-inequality}.
\begin{restatable}{claim}{SUDeltanorm}
    \label{clm:SUDelta-norm}
    For strictly convex space $\metric$, 
    $(|S| + |U|) \pam \leq \frac{2}{n} \sum_{i,j\in[n]} d_{ij}$. 
\end{restatable}

\begin{proof}
    The equilibrium condition from Claim~\ref{cl:pure_equilibrium} for each $i \in S \cup U$ gives us
    \be
    \label{eq:equilibrium_SU}
        t_i + \pam \le \cost_i(\rli) = \sum_{j \in [n]} \frac{1 / (d_j + \pam)}{D} c_{ij} = 
        \frac{|S|/\pam}{D} \cdot t_i + \sum_{j \in U\cup T} \frac{1 /(d_j + \pam)}{D} c_{ij}.
    \ee
    We recall that $D\eqdef\frac{|S|}{\pam}+\sum_{j\in U\cup T}\frac{1}{d_j+\pam}$, i.e., $\frac{|S|/\pam}{D}+\sum_{j\in U\cup T}\frac{1/(d_j+\pam)}{D}=1$. Hence, after subtracting $t_i$ from both sides of \eqref{eq:equilibrium_SU}, we get the following
    \begin{multline}
    \label{eq:equilibrium_SU_Delta}
        \pam\le\sum_{j\in U\cup T}\frac{1/(d_j+\pam)}{D}(c_{ij}-t_i)\overset{\text{Claim}~\ref{cl:norm-inequality}\text{ for }j\in U}{\le} \sum_{j\in T}\frac{1/(d_j+\pam)}{D}(c_{ij}-t_i)\\
        +\sum_{j\in U}\frac{1/(d_j+\pam)}{D}\cdot\left(\frac{t_j-d_j}{t_j}t_i-\frac{d_j}{t_j}d_{ij}-t_i\right)
        = \sum_{j\in T}\frac{d_{ij}-t_i}{D(t_j+\pam)}
        +\sum_{j\in U}\frac{d_{ij}-t_i}{D(d_j+\pam)}\cdot\frac{d_j}{t_j}\\
        \le\sum_{j\in T}\frac{1/(t_j+\pam)}{D}\cdot d_{ij}+
        \sum_{j\in U}\frac{1/(d_j+\pam)}{D}\cdot\frac{d_j}{t_j}
        \cdot d_{ij}\overset{\text{Lemma}~\ref{lem:prob-TU}}{\le}
        \sum_{j\in T}\frac{2}{n}\cdot d_{ij}+
        \sum_{j\in U}\frac{2}{n}\cdot\frac{t_j+\pam}{d_j+\pam}\cdot\frac{d_j}{t_j}\cdot d_{ij}\\
        \le\frac{2}{n}\sum_{j\in T\cup U}d_{ij},
    \end{multline}
    where to get the third inequality, we simply substituted $d_{ij}-t_i$ with $d_{ij}$; and to get the last inequality, 
    we used that $\frac{t_j+\pam}{d_j+\pam}\cdot\frac{d_j}{t_j}=\frac{1+\pam/t_j}{1+\pam/d_j}\le 1$. The claim follows after we add \eqref{eq:equilibrium_SU_Delta} for all $i\in S\cup U$.    
   \[
        (|S| + |U|) \pam = \sum_{i \in S \cup U} \pam \leq \frac{2}{n} \sum_{i \in S \cup U} \sum_{j \in U \cup T} d_{ij}\le 
        \frac{2}{n} \sum_{i,j\in[n]} d_{ij}.%
    \]
\end{proof}

We get the following bound on $(*)$ by applying Claim~\ref{clm:SUDelta-norm} to \eqref{eq:bound_star_norm}.
\[
(*)\le \frac{2}{n} \sum_{i,j\in[n]} d_{ij}\cdot\max\left\{ \frac{1 / c + 1}{n}, \frac{1 / c + 1}{n / c}\right\}\cdot\frac{n}{c}=
\max\Big\{ \frac{2 + 2c}{c^2}, \frac{2 + 2c}{c}\Big\} \cdot \frac{1}{n} \sum_{i,j\in[n]} d_{ij}.
\]
I.e., $(*)=O(1+\frac{1}{c^2})\cdot\opt$. We conclude the proof of Theorem~\ref{thm:harm-robust-normspace} by combining this bound on $(*)$ with the constant upper bound \eqref{eqn:star2-UB}  on $(**)$ in \eqref{eq:normed_U}.

\subsubsection{Robustness in General Metric Spaces}
\label{sec:robust_general}
The goal of this section is to extend our robustness result for strictly convex spaces to general metric (Theorem~\ref{thm:harm-robust-metricspace}). However, since 
general metric $\metric$ may have rather different geometry than strictly convex spaces\footnote{E.g., it may violate the inequality~\eqref{eq:fact_normed} from Fact~\ref{fact:norm-inequality}.
consider a two-dimensional space $L_1(\R^2)$, and consider a point $(0,1)\in\shortpath((0,0),(1,1))$. The distances from another point $(1,0)$ to the endpoints are $d((1,0),(0,0))=d((1,0),(1,1))=1$, while $d((1,0),(0,1))=2$.}, we lose an additional factor $1/c$ in our robustness guarantee.
\begin{restatable}{theorem}{robustmetricspace}
    \label{thm:harm-robust-metricspace}
    $\harm(\pam)$ is a $O(1+1/c^3)$-robust in metric space, if $\Delta=\frac{c}{n}\cdot\SC(\po,\ells)$.
\end{restatable}

\begin{proof} As in the proof of Theorem~\ref{thm:harm-robust-normspace}, the most technically challenging part is to handle agents $j\in U$ in the expression~\eqref{eqn:robust-SC} for the expected social cost of \harm \ mechanism. However, we no longer have the bound on $c_{ij}$ provided by Claim~\ref{cl:norm-inequality}, which makes it harder to deal with agents $j\in U$. We still follow similar approach and get an upper bound on $c_{ij}$ of the form $O(t_i) + O(\frac{d_j}{t_j}d_{ij})$ in our next Claim~\ref{clm:metric-approximation-c}. However, another subtle and more serious issue is that Claim~\ref{clm:SUDelta-norm} in the proof of 
Theorem~\ref{thm:harm-robust-normspace} uses Claim~\ref{cl:norm-inequality}, which the new Claim~\ref{clm:metric-approximation-c} cannot resolve. We derive instead another upper bound on $(|S|+|U|)\pam$ built upon different ideas in Claim~\ref{clm:metric-SUDelta} and also lose additional factor $c$ in our guarantee.
\begin{claim}
    \label{clm:metric-approximation-c}
    For any metric space $\metric$, $c_{ij} \leq 2 t_i + 2 \frac{d_j + \pam}{t_j + \pam} d_{ij}$ for any $j \in U$ and $i\in[n]$.
\end{claim}
\begin{proof} We prove a slightly stronger claim that $c_{ij} \leq 2 t_i + 2 \frac{d_j}{t_j} d_{ij}$, as $\frac{d_j}{t_j}\le\frac{d_j+\pam}{t_j+\pam}$. By triangle inequalities for points $\po,\rlj,\elli[j],$ and $\elli$ we have the following bounds on $c_{ij}=d(\elli,\rlj), t_i=d(\po,\elli), t_j=d(\po,\elli[j]), d_j=d(\po,\rlj),$ and $d_{ij}=d(\elli,\elli[j])$:
$$ c_{ij} \leq t_i + d_j \quad \text{and} \quad d_j \leq t_j \leq t_i + d_{ij}.$$
The proof proceeds by considering the following four simple cases. 
\begin{enumerate}
    \item If $t_i\ge t_j$, then $2t_i\ge t_i+t_j\ge t_i+d_j\ge c_{ij}$ and the claim follows.
    \item If $2d_j\ge t_j$, then $2t_i+2 \frac{d_j}{t_j} d_{ij}\ge
    2t_i+ d_{ij}\ge t_i + d_j\ge c_{ij}$ and the claim follows. 
    \item If $t_i<t_j$ and $\frac{2(t_j-t_i)}{t_j}\ge 1$. Then, since $d_{ij}\ge t_j - t_i$, we have $2t_i+2 \frac{d_{ij}}{t_j} d_j\ge 2t_i+2 \frac{t_j-t_i}{t_j} d_j\ge 2t_i + d_j\ge c_{ij}$ and we are done.
    \item If $t_i<t_j$, $\frac{2(t_j-t_i)}{t_j}< 1$, and $2d_j< t_j$. Then, as $\frac{2(t_j-t_i)}{t_j}< 1$, we have $t_j<2t_i$. The claim follows, as $2t_i> t_i + \frac{1}{2}t_j>t_i+d_j\ge c_{ij}$.   
\end{enumerate}
This concludes the proof of the claim, as if none of the first three cases hold, then all conditions of the last case are necessarily satisfied.  
\end{proof}
As before, we apply this upper bound on $c_{ij}$ for all $j\in U$ and $i\in[n]$ in the expected social cost~\eqref{eqn:robust-SC} and get
\begin{align}
    \label{eq:metric_U}
\SC &\le
\frac{|S|}{\pam\cdot D} \sum_{i \in [n]} t_i +
\sum_{j \in U} \frac{1 / (d_j + \pam)}{D} \sum_{i \in [n]} \Big( 2t_i + 2\frac{d_j+\pam}{t_j+\pam} d_{ij}\Big)+
\sum_{j\in T}\frac{1 / (t_j + \pam)}{D} \sum_{i \in [n]}d_{ij}\\
&=\underbrace{\left(\frac{|S|}{\pam\cdot D}+\sum_{j\in U}\frac{2}{D(d_j + \pam)}\right)\cdot\sum_{i \in [n]} t_i}_{(*)} + 
\underbrace{\sum_{j\in U}\frac{2/(t_j+\pam)}{D}\sum_{i \in [n]}d_{ij} +\sum_{j\in T}\frac{1/(t_j + \pam)}{D} \sum_{i \in [n]}d_{ij}}_{(**)}.\nonumber
\end{align}
We handle the term $(**)$ in almost exactly the same way as we did in \eqref{eqn:star2-UB} by applying Lemma~\ref{lem:prob-TU}.
\be
    \label{eqn:metric_star2-UB}
    (**) \le
    \sum_{j\in U}\frac{4}{n}\sum_{i\in[n]}d_{ij}+
    \frac{2}{n}\sum_{j\in T}\sum_{i\in[n]}d_{ij}\le
    \frac{4}{n}\sum_{j\in U\cup T}\sum_{i\in[n]}d_{ij}
    \le\frac{4}{n}\sum_{i,j\in[n]}d_{ij}\le 8\cdot\opt,     
\ee
where the last inequality holds by Lemma~\ref{thm:uniform-sampling-ratio} similar to Corollary~\ref{cor:contribution-T}. We also get a similar to \eqref{eq:bound_star_norm} upper bound on $(*)$ by using the definition of $c=\frac{\pam}{\pac}=\frac{\pam\cdot n}{\sum_{i\in[n]}t_i}$ and Lemma~\ref{lem:DeltaD}
\be
\label{eq:bound_star_metric}
    (*) \le
    \left(\frac{|S|}{\pam\cdot D}+\sum_{j\in U}\frac{2}{D\cdot \pam}\right)\cdot\sum_{i \in [n]} t_i \le
    2(|S|+|U|)\pam\cdot\max\left\{ \frac{1 / c + 1}{n}, \frac{1 / c + 1}{n / c}\right\}\cdot\frac{n}{c}.
\ee
We show a new upper bound on $(|S|+|U|)\pam$ in the following Claim~\ref{clm:metric-SUDelta} using only that $\metric$ is a metric space.
\begin{claim}
    \label{clm:metric-SUDelta}
    For any metric space $\metric$, $(|S| + |U|) \pam \le O(1+1/c) \opt$.
\end{claim}
\begin{proof} We first write equilibrium conditions from Claim~\ref{cl:pure_equilibrium} for each agent $i\in S\cup U$ and subtract $t_i$ from both sides as in \eqref{eq:equilibrium_SU}. I.e.,
    \begin{multline}
    \label{eqn:SU-pa-1}
        \pam \le \sum_{j \in U} \frac{1 / (d_j + \pam)}{D} (c_{ij} - t_i) + \sum_{j \in T} \frac{1 / (t_j + \pam)}{D} (d_{ij} - t_i)\\
        \overset{\text{Lemma}~\ref{lem:prob-TU}}{\le}\sum_{j \in U} \frac{1 / (d_j + \pam)}{D} (c_{ij} - t_i) +\frac{2}{n}\sum_{j\in T}d_{ij}.  
    \end{multline}
    We note that in the summation corresponding to agents $j\in U$ in the right hand side of \eqref{eqn:SU-pa-1} we can ignore negative terms $c_{ij}-t_i$. We also observe that $c_{ij}\le t_i$ when
    \begin{observation}
    \label{obs:large_dj}
        If $d_j \ge 2 d_{ij}$, then $c_{ij} \le t_i$.
    \end{observation}
    \begin{proof}
        By applying triangle inequality a few times we get $c_{ij} \leq d_{ij} + (t_j - d_j) \le 2d_{ij} + t_i - d_j$. Since $d_j \ge 2 d_{ij}$, we get $c_{ij} \leq t_i$.
    \end{proof}
    We continue with the bound \eqref{eqn:SU-pa-1} on $\pam$ using 
    Observation~\ref{obs:large_dj}.
    \begin{multline}
        \label{eq:metric_long_derivation}
        \pam \le \sum_{j \in U} \frac{1 / (d_j + \pam)}{D} \cdot \mathds{1}[2 d_{ij} \geq d_j] \cdot (c_{ij} - t_i) + \frac{2}{n} \sum_{j \in T} d_{ij}
        \le \sum_{j \in U} \frac{1 / (d_j + \pam)}{D} \cdot \mathds{1}[2 d_{ij} \geq d_j] \cdot d_j
        \\
         + \frac{2}{n} \sum_{j \in T} d_{ij}\le
         \sum_{j \in U} \frac{1 / (d_j + \pam)}{D} \cdot \mathds{1}\left[2 d_{ij} \ge d_j \ge \frac{\pam}{2}\right] \cdot d_j +
         \sum_{j \in U} \frac{1 / (d_j + \pam)}{D} \cdot \mathds{1}\left[d_j < \frac{\pam}{2}\right] \cdot d_j\\
         + \frac{2}{n} \sum_{j \in T} d_{ij}
         \le \sum_{j \in U}\frac{d_j}{(d_j + \pam)D}\cdot\mathds{1}\left[2 d_{ij} \ge d_j \ge \frac{\pam}{2}\right] + \frac{\pam}{2} + \frac{2}{n} \sum_{j \in T} d_{ij}
         ,
    \end{multline}
    where the second inequality holds, as $d_j+t_i\ge c_{ij}$ by triangle inequality; we get the third inequality by considering cases whether $d_j\ge\pam/2$ or not for $j\in U$; the last inequality holds, as the sum $\sum_{j\in U}\frac{1/(d_j+\pam)}{D}\le 1$ and $d_j<\pam/2$ whenever $\mathds{1}\left[d_j < \frac{\pam}{2}\right]> 0$. We get the following inequality after subtracting $\pam/2$ from both sides of \eqref{eq:metric_long_derivation} and estimating $\frac{d_j}{d_j+\pam}\le 1$
    \begin{multline}
        \label{eq:metric_half_delta}
        \frac{\pam}{2}\le
        \frac{1}{D}\cdot\sum_{j \in U}\mathds{1}\left[2 d_{ij} \ge d_j \ge \frac{\pam}{2}\right] + \frac{2}{n} \sum_{j \in T} d_{ij}\le
        \frac{4}{\pam D}\cdot\sum_{j \in U}\frac{\pam}{4}\cdot \mathds{1}\left[ d_{ij} \ge \frac{\pam}{4}\right] + \frac{2}{n} \sum_{j \in T} d_{ij}\\
        \le \frac{4}{\pam D}\cdot\sum_{j \in U}d_{ij}+ \frac{2}{n} \sum_{j \in T} d_{ij},
    \end{multline}
    where we get the second inequality by multiplying and dividing by $\frac{4}{\pam}$ the summation for $j\in U$ and also by simplifying the condition inside the indicator function; the last inequality follows as $d_{ij}\ge \frac{\pam}{4}\cdot \mathds{1}\left[ d_{ij} \ge \frac{\pam}{4}\right]$. Finally, we add inequalities \eqref{eq:metric_half_delta} for all $i\in S\cup U$ and get
    \begin{multline*}
        \frac{1}{2}(|S|+|U|)\pam=\sum_{i\in S\cup U}\frac{\pam}{2}\le\frac{4}{\pam D}\sum_{i\in S\cup U}\sum_{j\in U}d_{ij}+\frac{2}{n}\sum_{i\in S\cup U}\sum_{j\in T}d_{ij}\le\left(\frac{4n}{\pam D}+2\right)\frac{1}{n}\sum_{i,j\in [n]}d_{ij}\\
        \le\left(2+4n\cdot \max\Big\{ \frac{1 / c + 1}{n}, \frac{1 / c + 1}{n / c}\Big\}\right)\cdot 2\opt=O\left(1+\frac{1}{c}\right)\cdot\opt,
    \end{multline*}
    where we used Lemma~\ref{lem:DeltaD} and Lemma~\ref{thm:uniform-sampling-ratio} to get the last inequality.
\end{proof}
We combine \eqref{eq:bound_star_metric} with the above bound from Claim~\ref{clm:metric-SUDelta} to get that 
\[
(*)\le O\left(1+\frac{1}{c}\right)\opt\cdot
\max\left\{ \frac{1 / c + 1}{n}, \frac{1 / c + 1}{n / c}\right\}\cdot\frac{n}{c} = O\left(1+\frac{1}{c^3}\right)\opt.
\]
Theorem~\ref{thm:harm-robust-metricspace} follows after we combine this bound on $(*)$ with \eqref{eqn:metric_star2-UB} the bound $(**)\le O(\opt)$ in the upper bound~\eqref{eq:metric_U} on the expected social cost.
\end{proof}

\subsection{Constant Number of Agents}
\label{sec:constant_agents}
We conclude the PoA analysis of \harm \ mechanism with a special case of a small number of agents $n=O(1)$. In this regime, we are able to simplify the mechanism by setting 
the parameter $\pam=0$, i.e., we ignore the estimate $\pa$ for the average distance $\pac=\frac{1}{n}\sum_{i\in[n]}d(\elli,\po)$, and get the following  $\harm(0)$ mechanism.
\begin{algorithm}
\caption{Harmonic mechanism $\harm(0)$}\label{alg:harmonic_simple}
\KwData{Reported locations $\rls=(\rli[1],\ldots,\rli[n])$, prediction $\po\in\metric$}
\KwResult{Facility $f\sim\pi\{\rli[1],\ldots,\rli[n]\}$}
\lFor{$i\in[n]$}{
let $d_i= d(\rli, \po)$
}
Choose $\Prx{f\gets \rli}\eqdef\frac{1 / d_i}{\sum_{j=1}^n 1/d_j}$ for $i\in[n]$ 
\tcp*{$\Prx{f\gets \rli}\propto\frac{1}{d_i}$ (proportional to)}
\end{algorithm}

\begin{restatable}{theorem}{constantagents}
    \label{thm:constant_agents}
    $\harm(0)$ is $1$-consistent ($\gamma$-consistent for a $\gamma$-accurate prediction $\po$) in general metric $\metric$. It is $O(n)$-robust, when $\metric$ is a strictly convex space.
\end{restatable}

\begin{proof}
    The consistency guarantees immediately follow from Theorem~\ref{thm:consistency_poa}, when we set the parameter $\pam=0$ in $\harm(0)$. For the robustness guarantee, we first note that agents in $T$ contribute at most $4\cdot\opt$ by Corollary~\ref{cor:contribution-T} to the social cost. Hence, we only need to bound the contribution from agents in $S\cup U$. Let us fix any agent $i\in S\cup U$, the Claim~\ref{cl:pure_equilibrium} gives us the bound on $\pam=0$ as follows: $0=\pam\le cost_i(\po,\locsmi)- d(\elli,\po)=\sum_{j\in U\cup T\setminus i}\frac{1/d_j}{D(\rli=\po,\rlsmi)}(c_{ij}-t_i)$. This implies that there exists a $i^*\in U\cup T$, $i^*\ne i$ such that $d(\elli,\rli[i^*])=c_{ii^*}\ge t_i$. By Fact~\ref{fact:norm-inequality} and Claim~\ref{cl:norm-inequality} we get that $d_{ii^*}=d(\elli,\elli[i^*])\ge t_i$, since $\rli[i^*]\in\shortpath(\po,\elli[i^*])$ with $t_i=d(\elli,\po)$.
    Finally, we estimate the social cost of $\rli$ as the following. 
    \begin{multline*}
    \sum_{j\in[n]}d(\elli[j],\rli)\le
    \sum_{j\in[n]}(d(\elli[j],\elli)+t_i)\le
    n\cdot t_i+ (n-1)\cdot\max_{j,k\in[n]}d_{jk}\\
    \le
    n\cdot d_{ii^*}+ (n-1)\cdot\max_{j,k}d_{jk}
    \le (2n-1)\cdot\max_{j,k}d_{jk}\le (2n-1)\cdot\opt.    
    \end{multline*}
    Hence, the contribution to the social cost of agents $i\in S\cup U$ is not more than $O(n)\cdot \opt$ in $\harm(0)$. 
\end{proof}

\begin{remark}
    \label{rem:normed_vs_general}
    The assumption that $\metric$ is strictly convex space is crucial for deriving $O(1)$-robustness bound. This is because each agent $i\in U$ may increase the distance to everyone else compared to $\po$ and $\elli$. The analysis for sets $S$ and $T$ can be carried in the same way as for strictly convex spaces.  
\end{remark}
 To see this, 
    consider the following example of Nash equilibrium in a circle metric space $\metric$ with just $n=2$ agents with arbitrary large PoA. Specifically, let $d(\elli[1],\elli[2])=1$,
    $d(\po,\elli[1])=d(\po,\elli[2])=M$ while the length of the whole circle is $2M+1$ for a large constant $M\in\R$. Then the following $\rls$ is a Nash equilibrium: $\rli[1]$ and $\rli[2]$ are on their respective shortest paths from $\elli[1]$ and $\elli[2]$ to $\po$ with $d(\rli[1],\po)=d(\rli[2],\po)=0.5$. It is easy to verify that each of the agents $i=1$ and $i=2$ is indifferent between reporting $\rli=\po$ or $\rli=\elli$. I.e., by Claim~\ref{cl:pure_equilibrium} $\rls$ is a Nash equilibrium with the social cost $\SC(\rls)=2M$. The optimal location is, e.g., at $f=\elli[1]$ with the social cost of $1$.

\section{Conclusions}
\label{sec:conclusion}
In this paper we study a canonical problem of strategic single-facility location in general metric spaces under new lenses of ML augmented mechanism design framework. This framework not only allows to circumvent worst-case 
analysis limitations, but also enriches the design space of mechanisms in interesting new ways. It naturally led us to consider new type of 
non-truthful mechanisms (such as \harm \ mechanism) that have not appeared in the prior literature. We got useful insights about undominated strategies and equilibria 
structure for this mechanism. We proved that \harm \ mechanism has a $1+\eps$ price of anarchy bound when predictions are (nearly) accurate, while retaining a constant 
PoA of $O(1+1/\texttt{poly}(\eps))$ in the worst-case, when $\po$ is arbitrary bad (given that our mechanism's parameter $\Delta=\eps\cdot\SC(\po,\ells)$). 

Our PoA analysis of consistency and especially robustness significantly deviates from a typical PoA analysis, as (i) the general smoothness argument does not help in breaking $2$-approximation barrier 
and (ii) we have to use metric conditions in a non-trivial way (e.g., our PoA bounds are different for strictly convex and general metric spaces, when $n=O(1)$). The tightness of our PoA bounds for \harm \ remains an open problem. Another interesting open question is to find a mechanism that does not depend on a parameter $\pam$, but admits similar $(1+\eps)$-consistency and $O(1+1/\texttt{poly}(\eps))$-robustness PoA guarantees.

\

\noindent\textbf{Acknowledgments.} Nick Gravin's research is supported by National Key R \& D Program of China (2023YFA1009500), by NSFC grant 61932002, and by ``the Fundamental Research Funds for the Central Universities'' in China. Chen and Im are supported in part by  NSF grants CCF-1844939, CCF-2121745, and CCF-2423106. 

\newcommand{\etalchar}[1]{$^{#1}$}

\bibliographystyle{alpha}

\appendix

\section{Missing Proofs}

\fctnorminequality*
\begin{proof} \label{pf:triangle-strictly-convex-space}
By strict convexity and since $\norm{P_2-P}+\norm{P-P_1}=\norm{P_2-P_1}$, we have $P=\frac{\norm{P_2-P}}{\norm{P_2-P_1}}\cdot P_1 + \frac{\norm{P-P_1}}{\norm{P_2-P_1}}\cdot P_2$. Hence, 
$P-A=\frac{\norm{P_2-P}}{\norm{P_2-P_1}}\cdot (P_1-A) + \frac{\norm{P-P_1}}{\norm{P_2-P_1}}\cdot (P_2-A)$. Then by triangle inequality for the norm $\norm{\cdot}$, we get
\begin{multline*}
d(P,A)=\norm{P-A}=\norm{\frac{\norm{P_2-P}}{\norm{P_2-P_1}}\cdot (P_1-A) + \frac{\norm{P-P_1}}{\norm{P_2-P_1}}\cdot (P_2-A)}\\
\le\frac{\norm{P_2-P}}{\norm{P_2-P_1}}\cdot\norm{P_1-A}+
\frac{\norm{P-P_1}}{\norm{P_2-P_1}}\cdot \norm{P_2-A}=
\frac{d(P,P_2)}{d(P_1,P_2)}\cdot d(A,P_1)+
\frac{d(P,P_1)}{d(P_1,P_2)}\cdot d(A,P_2).
\end{multline*}
\end{proof}

\end{document}